\documentclass[pra,aps,floatfix,amsmath,superscriptaddress,twocolumn,longbibliography,nofootinbib]{revtex4-2}
\usepackage{amsmath,amssymb,amsthm,mathtools,braket,enumitem}
\usepackage{comment}
\usepackage[OT1]{fontenc}
\usepackage{lmodern}
\usepackage{times}
\usepackage{xcolor}
\usepackage{quantikz}
\usepackage{hyperref}
\hypersetup{colorlinks,linkcolor={red!50!black},citecolor={blue!50!black},urlcolor={blue!80!black}}
\usepackage[capitalise,nameinlink]{cleveref}
\newtheorem{theorem}{Theorem}

\newtheorem{conjecture}[theorem]{Conjecture}
\newtheorem{definition}[theorem]{Definition}
\newtheorem{lemma}[theorem]{Lemma}

\newcommand{\ketbra}[2]{\ket{#1}\!\bra{#2}}

\newcommand{\complexity}[1]{\ensuremath{\mathsf{#1}}}
\newcommand{\poly}{\operatorname{poly}}
\newcommand{\halfBQP}{\ensuremath{\frac12\mathsf{BQP}}}
\newcommand{\halfBQPone}{\ensuremath{\frac12\mathsf{BQP}_1}}
\newcommand{\Factoring}{\textnormal{\textsc{Factoring}}}
\makeatletter
\renewcommand\onecolumngrid{%
  \do@columngrid{one}{\@ne}%
  \let\set@footnotewidth\set@footnotewidth@one
  \def\footnoterule{\kern-6pt\hrule width 1.5in\kern6pt}%
}
\makeatother

\newcommand{\Z}{\mathbb Z}

\begin{document}
\title{Protecting Quantum Computers against Untrusted Users}
\author{Shiv Akshar Yadavalli}\affiliation{Joint Center for Quantum Information and Computer Science, University of Maryland and NIST, College Park, MD, USA}\author{Joel Rajakumar}\affiliation{Electrical \& Computer Engineering Department, The University of California, Los Angeles, CA, USA}\author{Alexander Schuckert}
\affiliation{DIENS, École Normale Supérieure, PSL University, CNRS, INRIA, 45 rue d’Ulm, Paris 75005, France}
\affiliation{QuEra Computing Inc., 1380 Soldiers Field Road, Boston, MA 02135, USA}
\affiliation{QuEra Computing UK Ltd., Unit 2, Tech Foundry 2, Didcot, OX11 0DF}
\author{Michael J. Gullans}
\affiliation{Joint Center for Quantum Information and Computer Science, University of Maryland and NIST, College Park, MD, USA}
\affiliation{QuEra Computing Inc., 1380 Soldiers Field Road, Boston, MA 02135, USA}
\begin{abstract}
Publicly accessible fault-tolerant quantum computers must preserve scientific utility while limiting cryptanalytic power. We propose the restricted model of computation, $\halfBQPone$: a quantum server provides random computational-basis inputs, revealed only after execution, and one designated output bit. This interface permits arbitrary circuits and system sizes. We conjecture that a classical client making polynomially many adaptive requests cannot efficiently factor RSA moduli. The model subsumes the well-known one-clean-qubit model $\complexity{DQC1}$, retaining its applications to infinite-temperature multi-time correlations, out-of-time-order correlators, and suitably normalized partition functions. We propose a candidate for separating $\halfBQPone$ from $\complexity{DQC1}$ based on testing classical predictions of quantum spin dynamics. We argue security in two steps. First, we show that one-bit readout makes the quantum stages of standard factoring algorithms classically simulable, even when their preparation and postprocessing are redesigned within a broad family around a single arithmetic-oracle call. This covers the quantum stages of Shor, Eker{\aa}--H{\aa}stad, Kitaev, and Regev factoring constructions. Second, we examine the workaround of coherently implementing rational reconstruction to release a factor bit. Random inputs obstruct this straightforward attack: known techniques accommodate logarithmic-depth classical circuits, but rational reconstruction has resisted such parallelization for decades. Therefore, we show that security of our protocol is built on not only classical hardness for factoring, but also conjectured circuit lowerbounds for rational reconstruction.
\end{abstract}

\maketitle

\section{Motivation}

Utility-scale quantum computers could arrive within this decade \cite{preskill2025megaquop,ibm2025starling}, and their effects will cascade through industry, governance, and the security assumptions that underwrite modern life. As a community of quantum scientists and engineers, we cannot defer these implications to an indefinite future; we must address them now.

Quantum computing profoundly expands our ability to understand complex quantum systems, with real implications for chemistry, clean energy, medicine, and material science \cite{bauer2020chemistry,reiher2017reactions,santagati2024drug}. However, the same capability can also be used to undermine digital trust: at sufficient scale, Shor's quantum factoring algorithm threatens to break widely deployed public-key cryptography that relies on the classical hardness of factoring large numbers  \cite{shor1997}. These cryptography protocols form the security foundation of digital communication, commerce, and software integrity globally \cite{nist2024transition}. Breaking them would erode the basic guarantees of privacy, authenticity, and integrity that keep the internet (and the institutions built on it) functioning.

While this makes post-quantum cryptography both necessary and inevitable \cite{nist2024transition}, this transition will take years and require substantial effort. Existing disparities in cybersecurity resources raise the concern that vulnerable organizations could be left behind~\cite{wef2025cyber}. Worse, the ``harvest now, decrypt later'' incentive for malicious state actors makes this a present-day governance problem, not a future one~\cite{nist2024transition}. Looking further ahead, untrusted users could also be misaligned AI agents with enormous classical computational resources, autonomously interfacing with quantum computing systems and directing them toward harmful objectives. We must  therefore \textit{disallow} the use of quantum computers to break deployed cryptography --- even as we allow ample power for beneficial applications like simulation and discovery.\\

We treat this challenge posed by quantum progress as a technical design problem. We develop simple interfaces that obstruct standard factoring algorithms and a broad family of adversarial variants, while preserving substantial computational power for scientific discovery.

\subsection{Public access and the security interface}
\label{sec:interface}

Imagine a world with several public fault-tolerant quantum computers that classical clients can access. A classical client submits polynomial size circuits to a quantum server to execute and receives classical replies. The client may make polynomially many adaptive requests and process them in randomized polynomial time.

In this setting, unrestricted quantum access permits Shor's factoring algorithm, and this would compromise RSA-based encryption and signatures. Progress in both hardware and algorithms makes this concern increasingly urgent. Recent experiments have demonstrated improved logical error correction~\cite{sivak2026reinforcement}, up to 94 error-detected and 48 error-corrected logical qubits have been encoded in trapped ions, with logical operations outperforming their unencoded counterparts~\cite{quantinuum2026iceberg}, and universal logical operations on  a neutral-atom platform, using arrays of up to 448 physical atoms~\cite{bluvstein2026fault}. Resource estimates for superconducting qubits have lowered the projected physical-qubit requirement for RSA-2048 from 20 million to below one million, with a runtime below one week in the latter estimate under specified hardware assumptions~\cite{gidney2021factoring,gidney2025factoring}. For reconfigurable qubits such as neutral atoms and trapped ions, a surface-code architecture exploiting reconfigurable connectivity and parallel (transversal) logical gates reduced the projected RSA-2048 runtime to 5.6 days using 19 million physical qubits with 1-ms error-correction cycles~\cite{zhou2025resource}, and a preliminary estimate based on qLDPC codes put RSA-2048 factoring at approximately 102,000 physical qubits and 97 days, assuming 1-ms error-correction cycles and highly parallel logical operations~\cite{cain2026shor}.

These advances bring closer a future in which unrestricted public quantum access puts cryptographically relevant factoring within reach of any user.Furthermore, an untrusted user can be exceedingly creative with their circuit and need not follow a prescribed algorithm or compilation method. They may disguise its purpose or distribute a computation across adaptive requests. Screening for recognizable cryptanalytic algorithms therefore leaves an obfuscation problem ~\cite{huang2026obfuscation}.

A naive way to disallow users from factoring RSA-2048 would be to restrict the total circuit size that a user can execute. However, this is unsatisfactory as we would like users to reap the benefits of larger and larger quantum computers as they become available. Thus, it would be ideal to identify some method to disallow factoring at \textit{all system sizes} (an asymptotic scale).

In this work, we propose a \textit{security interface} regulating the inputs and outputs of quantum computations specified by the client. The server controls input preparation and output measurement, as well as what information about the inputs and outputs is returned to the client and when. The client remains free to design the computation, including adversarially. Security must therefore concern the computational power of the full interaction, including classical processing and adaptive requests. If the interface restricts the full interaction to computations that cannot factor efficiently, disguising an algorithm cannot evade that restriction. These input and output rules define a \textit{restricted computational model}.

\subsection{Connection to restricted models of computation}

Restricted models of quantum computation have been studied extensively in theoretical computer science, motivated by the question of whether there is a meaningful notion of a computer with power intermediate between classical and quantum. This is exactly the kind of asymptotic restriction the previous section asks for, and the observation that drives this paper is that such models can often be imposed through the interface alone, by controlling nothing but inputs and outputs.

The canonical example is the one-clean-qubit model, $\complexity{DQC1}$~\cite{knill1998}, in which the computer can manipulate a single qubit in a known pure state while every other qubit is maximally mixed. Originally motivated by NMR, the model is surprisingly powerful: it performs computations believed to require classical time exponential in the number of \emph{maximally mixed} qubits. Crucially, no method is known for factoring within this complexity class. And a server can impose it without inspecting the client's circuit at all: the client submits an arbitrary circuit, and the server runs it on one clean qubit and returns one output bit.

This model would already serve our purpose, with plausible security against factoring and known applications. We go one step further and ask for the most powerful restricted model that still appears unable to factor. Our starting point is $\halfBQP$, a model recently shown to factor~\cite{jacobs2024half}; we find a simple weakening that removes this capability, and the result is $\halfBQPone$, which can alternatively be seen as a computational model with slightly more computational power than $\complexity{DQC1}$. As with most questions of this complexity-theoretic flavor, we cannot prove definitively that $\halfBQPone$ fails to factor, nor that it is strictly more powerful than $\complexity{DQC1}$ (for instance, no one has proved $\complexity{DQC1}$ strictly weaker than $\complexity{BQP}$ either). What we can do is say precisely what a proof of either statement would require, and give reasons to believe both.

\section{Overview}
\label{sec:overview}

In this work, we have translated our practical concern into a theoretical question,\\

\textit{What is the strongest restricted model of quantum computation that does not contain factoring?}\\

Our technical contributions primarily have to do with identifying such a candidate. We precisely identify the necessary assumptions one must make in order for such a restricted model of quantum computation to be incapable of factoring asymptotically and retain utility. We also remark briefly on how simpler security interfaces will allow enhanced security and enhanced utility in the near-term, where finite-size effects come into play.

\subsection{The Security Interface}

The interface we propose is the following (Fig.~\ref{fig:interface}a). The client fixes an $n$-qubit \textit{unitary} circuit $U$. Then \textit{after the circuit is fixed}, the server draws a uniform string $x\in\{0,1\}^n$, prepare $U\ket x$, measures the \textit{first qubit} in the computational basis, and returns the outcome together with $x$. We call this model $\halfBQPone$, the subscript recording that a single bit is read out; Section~\ref{sec:complexity} defines it formally alongside $\complexity{DQC1}$ and $\halfBQP$ (Table~\ref{tab:models}), along with their associated complexity classes.

\begin{figure}[t]
\centering
\small
(a)\quad
\begin{quantikz}[row sep={0.55cm,between origins}, column sep=0.5cm]
\lstick[wires=3]{$\ket{x}$, $x$ uniform} & \gate[wires=3]{U} & \meter{} & \rstick{$y$} \qw \\
 & \qw & \qw & \qw \\
 & \qw & \qw & \qw
\end{quantikz}
\\[1.5em]
(b)\quad
\begin{quantikz}[row sep={0.55cm,between origins}, column sep=0.5cm]
\lstick[wires=3]{$\ket{s}$} & \gate[wires=3]{V} & \gate[wires=3]{O_\pi} & \gate[wires=3]{\mathcal N} & \meter{} & \rstick{$z$} \qw \\
 & \qw & \qw & \qw & \qw & \qw \\
 & \qw & \qw & \qw & \qw & \qw
\end{quantikz}
\caption{(a) The $\halfBQPone$ interface. The client submits an $n$-qubit circuit $U$; only then does the server draw a uniform $x\in\{0,1\}^n$, prepare $\ket x$, apply $U$, measure the first qubit, and return the reply $(x,y)$. $\complexity{DQC1}$ differs only in fixing the first input qubit to $\ket0$ and returning $y$ alone; $\halfBQP$ differs only in measuring every qubit. (b) A single-call circuit $A=\mathcal N\,O_\pi\,V$ (Section~\ref{sec:query-class}): a tractable preparation $V$, one call to an arithmetic oracle $O_\pi$ such as modular exponentiation, and a normalizer circuit $\mathcal N$ such as a Fourier transform. With all qubits read, Shor's circuit has this form and factors; Theorem~\ref{thm:query-class-sampler} shows that with one bit read, every such circuit is classically simulable on every basis input $s$. The decoder attack of Section~\ref{sec:decoder-attack} would append a Boolean decoder between $\mathcal N$ and the measurement. The random input denies the clean workspace required by this attack.}
\label{fig:interface}
\end{figure}
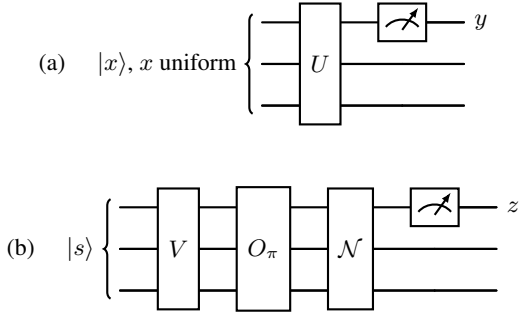

\subsection{Investigating Security at Asymptotic Scale}

Note all three of the input, output, and unitary restrictions in $\halfBQP_1$ are \textit{necessary} for security. Drop the randomization and the client has all of $\complexity{BQP}$, which can factor large numbers. Measure all $n$ qubits instead of one and the client has $\halfBQP$, which can also factor large numbers ~\cite{jacobs2024half}. Allow for mid-circuit measurements or reset and the client can factor via an old protocol of \cite{ParkerPlenio2000}.
Our central conjecture is that these restrictions may also be \textit{sufficient} for security.
\begin{conjecture}[Security of $\halfBQPone$]
\label{conj:main}
\[
 \boxed{\Factoring\notin
        \complexity{BPP}^{\complexity{Samp}\halfBQPone}}
\]
\end{conjecture}
\noindent The superscript denotes polynomially many adaptive requests to the $\halfBQPone$ interface, with randomized polynomial-time classical processing between them, and is made precise in Section~\ref{sec:complexity}; $\Factoring$ is the decision problem of returning one bit of a factor, made precise in Section~\ref{sec:security}. Note, since the classical client may ignore the server, the conjecture presupposes that factoring is classically hard. In this work, we investigate what further assumptions are necessary to say that the restricted quantum computation adds no additional ability to factor.

Our investigation of the conjecture comes in two steps, one per restriction. The first concerns the output. The quantum stage of each standard factoring algorithm---Shor, Eker{\aa}--H{\aa}stad, Kitaev, Regev---makes a single call to an arithmetic oracle such as modular exponentiation and its advantage lives in the full Fourier sample. We show in Theorem~\ref{thm:query-class-sampler}, Section~\ref{sec:query-class} that one bit of that sample carries none of it: with one-bit readout, every circuit of this shape is classically simulable.

We further show that this classical simulability is maintained over a broad family of post-processing steps (`normalizer' circuits) within the quantum circuit, suggesting that a successful attack on the security conjecture will likely need to leave the single-call framework.

The second step explains why the random input is still needed. In $\complexity{BQP}$, Shor's algorithm decides a bit of a factor by running its classical decoder coherently on clean ancillas after the Fourier transform, so that the one bit read out is already a bit of the factor. However, the decoder needs initialized workspace, and that is exactly what the random input denies. We investigate the possibility of coherently implementing a decoder with small workspace and we show that this is connected to a long-standing open problem of implementing \textit{rational reconstruction} with logarithmic-depth classical circuits. In particular, we prove explicitly in Theorem~\ref{thm:rr-attack}, Section~\ref{sec:decoder-attack} that by making repeated oracle calls (thus evading the no-go of Theorem~\ref{thm:query-class-sampler}) and also implementing rational reconstruction in logarithmic depth, one can engineer a successful factoring protocol in $\halfBQPone$. Rational reconstruction has resisted parallelization to logarithmic depth for decades, which is the reason to expect this escape to be closed; but since it lies in $\complexity P$, a proof of the conjecture would separate $\complexity P$ from $\complexity{NC}^1$.

We therefore do not expect an security theorem (since it would resolve a major open problem in circuit complexity). Our work instead identifies exactly how attempts to factor in $\complexity{\halfBQPone}$ will need to depart from known techniques: by leaving normalizer post-processing, by making multiple oracle calls, or by implementing rational reconstruction with low clean space. Moreover, we have shown that in addition to assuming that classical factoring is hard, our security conjecture requires that space/depth-restricted classical computation of continued fractions are impossible.

\subsection{Investigating Utility at Asymptotic Scale}

On the utility side, $\halfBQPone$ subsumes $\complexity{DQC1}$, so every simulation application of the one-clean-qubit model carries over. However, this is a large price compared to $\complexity{BQP}$. Any computation that needs a polynomial number of initialized qubits is unavailable in $\halfBQPone$ asymptotically: ground-state preparation and energy estimation, most of quantum chemistry, amplitude amplification. What survives is infinite-temperature physics, where random product states are not a nuisance but the natural starting point~\cite{heitmann2020typicality,rosenberg2024heisenberg}: transport and relaxation, scrambling, spectral statistics, thermodynamics at additive precision. However, near-term machines need not pay this price in full, as discussed in the next section (\cref{sec:overview-finite}).

The revealed input additionally allows the measurement outcome on the output qubit to be weighted by a classical function of the initial bitstring, i.e. an additional computational ability over $\complexity{DQC1}$. We propose a candidate application of this in Section~\ref{sec:classical-predictions}. We further note in Section~\ref{sec:security-simulation} proving $\halfBQPone$ is strictly more powerful than $\complexity{DQC1}$ would also encounter the same $\complexity P\ne\complexity{NC}^1$ barrier that proving security does, for similar reasons.

\subsection{Security and Utility at Finite Scale}\label{sec:overview-finite}
\begin{figure*}[ht]
  \centering
  \includegraphics[width=\textwidth]{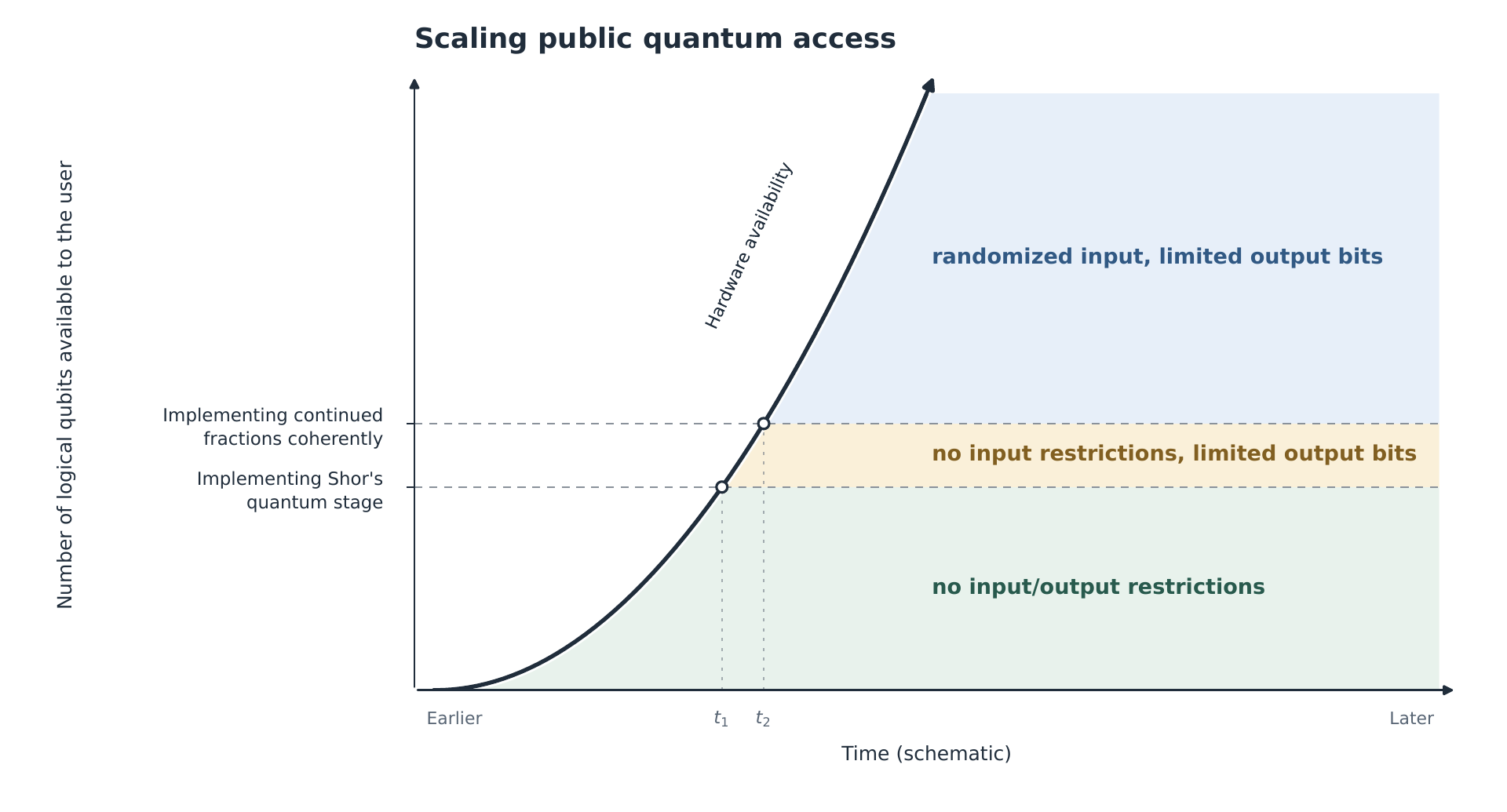}
  \caption{Schematic progression of our proposed interfaces for public quantum access as hardware scales, for a fixed modulus size. Each stage specifies the restrictions we argue suffice to prevent order-finding. \textbf{Green:} below the size required for Shor's quantum stage, so inputs and outputs can remain unrestricted. \textbf{Yellow:} Shor's quantum stage fits but the coherent continued-fraction decoder does not, so inputs can remain unrestricted while readout is limited to one bit. \textbf{Blue:} once the decoder also fits, the initial bitstring is sampled uniformly and revealed after execution, with readout still limited to one bit, giving the $\halfBQPone$ interface.}
  \label{fig:access-regimes}
\end{figure*}

Note, crucially, asymptotic statements about $\halfBQPone$ do not require all of the input qubits to be completely randomized: as long as the number of `clean' input bits do not scale with system size, i.e., a constant number, the model of computation is asymptotically equivalent to $\halfBQP_1$. This provides an important degree of freedom in practice as the server could choose this parameter in a manner that restricts known factoring protocols while taking advantage of scaling hardware.

In particular, \cref{thm:query-class-sampler} and \cref{thm:rr-attack} together suggest a finite-size security regime that does not require input randomization. Let $M_1$ be the number of logical qubits needed for the quantum stage of the most qubit-efficient known factoring algorithm, and let $M_2$ be the number of clean qubits needed to run a coherent decoder on the output of such a quantum stage. If there is a considerable gap between $M_2$ and $M_1$, a machine with $M$ clean input qubits, an unbounded number of dirty/random input qubits, and one-bit readout will be secure against known attacks when $M_1\le M<M_2$: Shor's stage fits, but without the decoder its one output bit is useless by \cref{thm:query-class-sampler}. Determining a reasonable lower bound on $M_2$ would require in-depth resource estimates, which is out of the current scope of this work. \Cref{sec:finite-size} shows that a straightforward coherent continued-fraction decoder costs about $10n$ logical qubits for an $n$-bit modulus, several times the arithmetic stage, and that optimized reversible Euclidean circuits bring this down; so $M_1$ could exceed $M_2$ by a considerable factor. This is the yellow region of \cref{fig:access-regimes}. In this regime one-bit readout retains every simulation task whose circuit fits the hardware, including tasks with clean inputs such as energy estimation~\cite{bauer2020chemistry}, and tasks with polynomially many numerical outputs at inverse-polynomial precision cost only polynomial measurement overhead~\cite{shi2024hydrodynamics,roymoulik2024correlations}.

\section{Complexity Theory for Public Access to Restricted Computational Models}
\label{sec:complexity}
First, we will distinguish between the model of computation itself and its associated complexity classes. To clarify to the non-complexity theorist, a model of computation generally indicates the properties and constraints imposed on a computer's inputs, outputs and computation executions, whereas its complexity class describes the problems that we can, or expect to solve using said computer. We use the same symbol, set in sans serif, for a named model and for its decision class, e.g.\ $\halfBQPone$; the words `model' or `interface' and `class' indicate which is meant.

\subsection{Restricted Models of Computation}

\begin{definition}[$\halfBQPone$]
    $\halfBQPone$ defines a model of computation which samples $x$ uniformly from \(\{0,1\}^n\), applies any polynomially-sized unitary circuit \(U\) to \(|x\rangle\), measures the first qubit to obtain \(y \in \{0,1\}\), and outputs \((x,y)\).
\end{definition}

\begin{definition}[$\halfBQP$]
    $\halfBQP$ defines a model of computation which samples $x$ uniformly from \(\{0,1\}^n\), applies any polynomially-sized unitary circuit \(U\) to \(|x\rangle\), measures all qubits to obtain \(y \in \{0,1\}^n\), and outputs \((x,y)\).
\end{definition}

\begin{definition}[$\complexity{DQC1}$]
    $\complexity{DQC1}$ defines a model of computation which samples $x$ uniformly from \(\{0,1\}^{n-1}\), applies any polynomially-sized unitary circuit \(U\) to \(\ket{0} \otimes |x\rangle\), measures the first qubit to obtain \(y \in \{0,1\}\), and outputs \(y\).
\end{definition}

\begin{table*}[t]
\caption{The three restricted models of this section, with unrestricted $\complexity{BQP}$ for comparison. In each case the client supplies an $n$-qubit circuit $U$, and the server controls the input, the measurement, and what is returned.}
\label{tab:models}
\begin{ruledtabular}
\begin{tabular}{llllll}
Model & Input state & Measured & Returned to client & Known to factor? & Simulation targets \\
\hline
$\complexity{BQP}$ & $\ket{0^n}$ & all qubits & $y\in\{0,1\}^n$ & yes~\cite{shor1997} & unrestricted \\
$\halfBQP$ & $\ket x$, $x$ uniform & all qubits & $(x,y)$, $y\in\{0,1\}^n$ & yes~\cite{jacobs2024half} & weighted transition statistics \\
$\halfBQPone$ & $\ket x$, $x$ uniform & first qubit & $(x,y)$, $y\in\{0,1\}$ & conjectured not (Conjecture~\ref{conj:main}) & weighted normalized traces (Sec.~\ref{sec:classical-predictions}) \\
$\complexity{DQC1}$ & $\ket0\otimes\ket x$, $x$ uniform & first qubit & $y\in\{0,1\}$ & no known method & normalized traces, OTOCs, partition functions \\
\end{tabular}
\end{ruledtabular}
\end{table*}

\subsection{Existing definitions of Decision Classes and Complete Problems}
One can now consider the natural notion of a decision problem that can be solved via \textit{a single request} to the restricted computational model. There are various inequivalent definitions in the literature\footnote{For instance, the classical controller can generate the circuit description in different ways.}; here, we use the definition of Jacobs and Mehraban~\cite{jacobs2024half}, which corresponds most closely to our setting.

\begin{definition} [Decision Problems]
    For a restricted computational model $C$, such as the three above, the complexity class $\complexity{C}$ is the set of languages $L \subseteq \{0,1\}^*$ for which, given any $z \in \{0,1\}^*$, a probabilistic polynomial-time
classical algorithm generates a circuit description from
$z$, makes \textnormal{one} request to $C$ with it, applies a polynomial-time classical function to the reply, and if $z \in L$ accepts with probability $\geq \frac12+\frac{1}{\poly(|z|)}$ and otherwise rejects with probability $\geq \frac12+\frac{1}{\poly(|z|)}$.
\end{definition}

Complete problems have been defined and studied for the decision classes $\complexity{DQC1}$ and $\halfBQP$. A complete problem captures the class exactly: it is solvable with a single request, and every problem solvable with a single request reduces to it classically. Formally, completeness concerns the gapped promise problems of deciding whether each quantity below exceeds a threshold or falls short of it by an inverse-polynomial margin. The following estimation problems characterize the three above
decision classes. In each case, the task is to estimate the
specified quantity to inverse-polynomial additive error with
success probability at least $2/3$. Here $U$ is a supplied
$n$-qubit unitary circuit, and each weight function $f$ is
specified by a classical circuit producing finite-precision
values in $[-1,1]$.

\begin{enumerate}
    \item \textit{Normalized trace estimation}
    ($\complexity{DQC1}$):
    estimate
    \[
        \frac{\operatorname{Re}\operatorname{Tr}(U)}{2^n}.
    \]

    \item \textit{Weighted normalized trace estimation}
    ($\halfBQPone$):
    given $f:\{0,1\}^n\to[-1,1]$, estimate
    \begin{equation}
        \frac{1}{2^n}
        \sum_{x\in\{0,1\}^n}
        f(x)\operatorname{Re}\langle x|U|x\rangle.
        \label{eq:weighted-trace-definition}
    \end{equation}

    \item \textit{Weighted transition statistic estimation}
    ($\halfBQP$):
    given $f:\{0,1\}^n\times\{0,1\}^n\to[-1,1]$, estimate
    \[
        \frac{1}{2^n}
        \sum_{x,y\in\{0,1\}^n}
        f(x,y)|\langle y|U|x\rangle|^2.
    \]
\end{enumerate}

Physically, Eq.~\eqref{eq:weighted-trace-definition} is an infinite-temperature average: the return amplitude of $U$ is averaged over a uniformly random initial product state $\ket x$, with each initial configuration weighted by the classical function $f(x)$. Replacing $U$ by a Hadamard-test circuit for an observable turns it into a weighted infinite-temperature expectation value, the form developed in Section~\ref{sec:classical-predictions}.

\subsection{Sampling Problems and Public Access}
While the prior definitions allow us to place $\halfBQPone$ relative to prior work, we remark that they are insufficient for our purposes. In particular, we would like the client to receive the raw replies of the interface and to process them over \textit{polynomially many adaptive} requests. In the setting of $\complexity{DQC1}$, the distinction between post-processing on replies and post-processing one output bit does not exist. However, for $\halfBQP$ and $\halfBQPone$, we would like to isolate what the interface itself returns from what the client does with it. To this end, we define \textit{sampling problems}: a sampling problem specifies a target output distribution for each input and the task is to sample from it to a prescribed accuracy. We refer the reader to Aaronson~\cite[Sec.~2.1]{aaronson2010sampling} for the more formal definition. We define it informally below,

\begin{definition}[$\mathsf{Samp}C$]
    For a restricted computational model $C$, $\mathsf{Samp}C$ is the set of
    families of distributions $\{D_z\}_{z \in \{0,1\}^*}$ for which there exists a
    probabilistic polynomial-time classical algorithm that, given $z$ and an
    accuracy parameter $\varepsilon > 0$, in time polynomial in $|z|$ and $1/\varepsilon$
    generates a circuit description and submits it \textnormal{once} to the
    model $C$, such that the output distribution of $C$ is within total variation
    distance $\varepsilon$ of $D_z$.
\end{definition}

In particular, problems in $\complexity{SampDQC1}$ correspond to probability distributions over $1$-bit outcomes, $\complexity{Samp}\halfBQPone$ to $(n+1)$-bit outcomes, and $\complexity{Samp}\halfBQP$ to $2n$-bit outcomes. Since the circuit description is generated by a probabilistic polynomial-time classical algorithm, these classes do not exclude classical computation: $\complexity{SampDQC1}$, for instance, contains every classically sampleable one-bit distribution, because the client can compute a bit $b$ classically and submit the circuit $X^b$ acting on the clean qubit. The purpose of the definition is rather to separate the raw replies of the interface from the client's processing of them, which the following class accounts for explicitly.

\begin{definition}[$\complexity{BPP}^{\complexity{Samp}C}$]
    For a restricted computational model $C$, $\complexity{BPP}^{\complexity{Samp}C}$ is the set of decision problems solvable, with success probability at least $2/3$ on every instance, by a probabilistic polynomial-time classical client that, polynomially many times and adaptively, submits a circuit description to the model $C$ and receives one reply drawn from its output distribution. In complexity-theoretic terms this is $\complexity{BPP}$ relative to an oracle for $\complexity{Samp}C$.
\end{definition}

With this notation, $\complexity{BPP}^{\complexity{Samp}\halfBQPone}$ is the class of decision problems solvable by a classical client with access to the $\halfBQPone$ interface, and it is the class that our security conjecture concerns.

We can also state a preliminary observation
\begin{theorem}
\label{thm:containments}
    \begin{align}
        \complexity{BPP}^{\complexity{SampDQC1}} \subseteq \complexity{BPP}^{\complexity{Samp}\halfBQPone} \subseteq \complexity{BPP}^{\complexity{Samp}\halfBQP} \ .
    \end{align}
\end{theorem}
\begin{proof}
To simulate one $\complexity{DQC1}$ request, the client makes one $\halfBQPone$ request of the same circuit, computes the output $y\oplus x_1$, and ignores the remaining bits $x_2\ldots x_n$. Since ignoring $x_1$ should result in the maximally mixed state, the runs with $x_1=1$ must be exactly the complement of the runs with $x_1 = 0$, so flipping $y$ whenever $x_1=1$ recovers the output corresponding to $x_1 =0$. To simulate one $\halfBQPone$ request, the client simply discards the last $n-1$ bits of the reply $y$ of a $\halfBQP$ request.
\end{proof}

\section{Security against factoring for $\halfBQPone$}
\label{sec:security}
Here, we define the factoring problem and state our main assumption.
\begin{definition}[Factoring]
    We define $\Factoring$ as the following decision problem. Given an $n$-bit semiprime $N$ and an integer $k \in [n]$, return the $k^{th}$ bit of the smaller prime factor of $N$, written with $n$ bits.
\end{definition}
Our conjecture, Conjecture~\ref{conj:main}, is that $\Factoring\notin\complexity{BPP}^{\complexity{Samp}\halfBQPone}$.
We could have considered the weaker conjecture that $\Factoring \not \in \complexity{BPP}^{\complexity{Samp}\complexity{DQC1}} $, but our purpose in proposing $\complexity{BPP}^{\complexity{Samp}\halfBQPone}$ is to start from $\Factoring \in \complexity{BPP}^{\complexity{Samp}\halfBQP} $ and find the strongest computational model that doesn't contain factoring. Thus, we isolate the ingredients necessary to factor.

We first show that one-bit readout removes the advantage of single-call factoring circuits in the near-term regime, then examine the asymptotic obstacle to implementing their classical decoder in randomized quantum workspace.

\subsection{Restricting the readout stops existing factoring algorithms out-of-the-box}
\label{sec:query-class}

The quantum stage of each standard quantum factoring algorithm has the following shape: Exponent
registers are prepared in an easy state, uniform or Gaussian; one
reversible arithmetic circuit computes a modular product into a work
register; the exponent registers are Fourier transformed; everything
is measured. Jacobs and Mehraban~\cite{jacobs2024half} observed that
this already allows factoring in $\halfBQP$. For Shor's circuit written with
Fourier transforms on both sides,
\begin{align}\label{eq:query-jacobs-mehraban}
 &A_{\mathrm{OF}}=(F_Q\otimes I)\,O_f\,(F_Q\otimes I),\\
 &O_f\ket{x,v}=\ket{x,\,v\oplus(a^x\bmod N)},
\end{align}
a random input $(u,v)$ only shifts the measured Fourier sample $y$
by the revealed $u$, so $u+y\bmod Q$ is an ordinary Shor sample and
the client decodes it classically. We now show that reading one bit
instead of $y$ removes the advantage of every circuit of this shape,
on every input, however the preparation and the postprocessing are
redesigned within the following family.

A \emph{single-call circuit} on $m$ qubits is
\begin{equation}\label{eq:query-forward-family}
 A=\mathcal N\,O_{\pi}\,V ,
\end{equation}
specified by three classical descriptions (Fig.~\ref{fig:interface}b). Firstly, $\mathcal N$ is a normalizer circuit~\cite{normalizer}, the generalization of a Clifford circuit from $\Z_2^m$ to an encoded group
$G=\prod_i\Z_{2^{t_i}}$ with $\sum_it_i=m$, given as a gate list; this class contains all
Clifford circuits and all products of cyclic Fourier transforms. Secondly, $O_\pi\ket g=\ket{\pi(g)}$ for a permutation $\pi$ of the basis
states, given as classical circuits for $\pi$ and $\pi^{-1}$: either an idealized arithmetic oracle such as $O_f$
above, granted to the client whether or not it could be implemented
without clean workspace, or any classical reversible (Toffoli) network for the arithmetic, whose action on every basis state, workspace included, is such a permutation. Finally, $V$ is a \emph{tractable preparation}:  any circuit supplied with classical algorithms that, for every basis input $s$, sample $V\ket s$ in the computational basis and evaluate its amplitudes, in the sense of Van den Nest's
computationally tractable states~\cite{vandenNest2011probabilistic}. This can, for instance, be a normalizer circuit over $G$, or a product of explicit circuits on $O(\log m)$-qubit blocks and single-qubit gates. All running times below are polynomial in the total description length.

This family contains the quantum stages of
Shor and Jacobs--Mehraban, Eker{\aa}--H{\aa}stad, Kitaev with deferred
measurements, and Regev with its space-optimized variants
(this is shown in Appendix~\ref{sec:readout-arithmetic}). It is interesting to note that the generalization beyond normalizer preparations for tractable preparations is only required to include the Regev construction; the remaining constructions sandwich the single call of $O_\pi$ with normalizer circuits. This separation of arithmetic access from the surrounding normalizer operations has been used before for analysing Shor's order-finding algorithm over black-box groups~\cite[Sec.~5.2]{bermejo2014blackbox}.

\begin{theorem}[One-bit readout of single-call circuits is classical]
\label{thm:query-class-sampler}
Let $A$ be a \textnormal{single-call} circuit on $m$ qubits, as in Eq.~\eqref{eq:query-forward-family}, and $s \in \{0,1\}^m$ any computational-basis
input. A probabilistic polynomial-time classical algorithm, given the description of $A$, samples the measurement outcome of the first qubit of $A\ket s$ to total variation error $\varepsilon$ in time polynomial in the description length and $1/\varepsilon$. Consequently, in $\halfBQPone$ the complete reply $(s,z)$ of such a request, and the complete transcript of $T$ adaptive such requests, are classically sampleable to total variation error $\delta$ in time polynomial in the input length, $T$ and $1/\delta$.
\end{theorem}
Note that sampling from the measurement distribution of all qubits of $A_{\mathrm{OF}}\ket s$ should be classically hard, since such samples can be processed into a factor, and we assume $\Factoring \not \in \complexity{BPP}$.

The proof sketch proceeds as follows: the arithmetic permutation preserves tractability, so $\ket{\phi_s}=O_\pi V\ket{s}$ remains efficiently sampleable with efficiently computable amplitudes.
Writing $Z_1=\sum_\chi c_\chi D_\chi$ in character observables gives
\[
\beta(s)
=\bra{\phi_s}\mathcal N^\dagger Z_1\mathcal N\ket{\phi_s}
=\sum_\chi c_\chi
\underbrace{\bra{\phi_s}M_\chi\ket{\phi_s}}_{\text{classically estimable}},
\]
where each $M_\chi=\mathcal N^\dagger D_\chi\mathcal N$. For $G$, $\chi_k(g)$ are its Fourier characters, labeled by $k\in G$, and the corresponding diagonal unitary is $D{\chi_k}=\sum_{g\in G}\chi_k(g)\ket{g}\bra{g}$. Therefore, $M_\chi$ is an efficiently computable monomial unitary, making its expectation classically estimable by Van den Nest's sparse-operator result~\cite[Theorem~3]{vandenNest2011probabilistic}.
Since the character terms can be sampled efficiently and $\sum_\chi|c_\chi|=O(m)$, their weighted sum is also efficiently estimable, yielding a classical simulator for the released bit and for polynomially many adaptive requests. See Appendix~\ref{sec:readout-sampler} for detailed proof.

\subsection{Randomizing the input addresses the decoder attack}
\label{sec:decoder-attack}

Theorem~\ref{thm:query-class-sampler} concerns circuits whose
postprocessing is a normalizer, and in $\complexity{BQP}$ nothing
forces an adversary to stop there. If a Boolean decoder $f$ is applied
to the Fourier register before the bit is read, the observable becomes
$\mathcal N^\dagger D_f\mathcal N$ with $D_f=\sum_y(-1)^{f(y)}\ketbra yy$.
Each monomial expectation is still easy to estimate, but the total coefficient weight can be exponential. The proof therefore no longer guarantees efficient estimation of their sum.
This is exactly how Shor's algorithm decides a bit of a factor: with
clean ancillas the client appends the continued-fraction decoder
reversibly after the Fourier transform within the quantum computation, and measures one bit of its
output. What the
decoder needs, however, is initialized
workspace, and the random input in $\halfBQPone$ renders all qubits effectively maximally mixed (until the classical post-processing step).

However, we still have to be careful here: Barrington's theorem allows us to evaluate any logarithmic-depth classical circuit reversibly
with limited workspace, and logarithmic-depth classical computation covers a variety of common arithmetic circuits. Two obstacles then stand
between an attacker and a coherent decoder of this kind: (1) the decoders of the standard
algorithms, continued fractions for the Shor-type algorithms
and lattice reduction for the Regev-type ones, are not known
to be implementable in logarithmic depth; rational reconstruction in particular has been studied extensively and resisted prior attempts at parallelization, (2) the decoder cannot be applied directly to the output of Shor's since it has an unknown offset resulting from the random input. In the appendix, we show how to handle this second obstruction, and so one needs to assume that rational reconstruction is not parallelizable in order for security.

To make this precise, given a dyadic number $t\in[0,1)$, let $\mathrm{RR}(N,t)$ be the
denominator of the reduced fraction $c/d$ satisfying
\[
    0\le c<d<N,
    \qquad
    \left|t-\frac cd\right|<\frac{1}{2N^2},
\]
or zero if no such fraction exists. There is at most one such
fraction: two distinct fractions with denominators below $N$ are more
than $1/N^2$ apart. The assumption $\mathrm{RR}\in\complexity{NC}^1$
means that each output bit of this function has a uniform
polynomial-size Boolean circuit of logarithmic depth.

\begin{theorem}[rational reconstruction attack]
\label{thm:rr-attack}
$\mathrm{RR}\in\complexity{NC}^1\implies
 \Factoring\in\complexity{BPP}^{\complexity{SampDQC1}}
 \subseteq\complexity{BPP}^{\complexity{Samp}\halfBQPone}$.
\end{theorem}

Since $\mathrm{RR}\in\complexity P$, a security proof
$\Factoring\notin\complexity{BPP}^{\complexity{Samp}\halfBQPone}$
would imply $\mathrm{RR}\notin\complexity{NC}^1$ and hence
$\complexity P\ne\complexity{NC}^1$. We therefore do not expect a
security theorem for $\halfBQPone$ without assumptions on the algorithms
an attacker may use. The proof of Theorem~\ref{thm:rr-attack} (Appendix~\ref{sec:rr-attack}) uses a
circuit that calls the arithmetic oracle polynomially many times, separated
by Fourier transforms on the control register, so it is not a
single-call circuit. This lines up with Theorem~\ref{thm:query-class-sampler}, which shows that a successful attack cannot be realized within the single-call family.

\section{Quantum simulation applications in $\halfBQPone$}

As mentioned, $\halfBQPone$ includes the well known one clean qubit model, $\complexity{DQC1}$. This ensures that all the progress made regarding simulation in $\complexity{DQC1}$ ports to $\halfBQPone$ -- this includes normalized traces and broadened spectral densities~\cite{knill1998};
average fidelity decay and Loschmidt echoes~\cite{poulin2004fidelity};
infinite-temperature multi-time correlations and out-of-time-order
correlators (OTOCs)~\cite{roymoulik2024correlations}; and suitably
normalized partition functions to additive
precision~\cite{chowdhury2021partition}. These correlators are of technical and societal importance, covering spectral functions measured in neutron-scattering experiments~\cite{scheie2021kpz}, the free-induction decay measured in nuclear magnetic resonance~\cite{lowe1957fid,seetharam2023nmr} and spin transport coefficients~\cite{kubo1957,bertini2021transport,joshi2022hydrodynamics} important in spintronics~\cite{zutic2004spintronics}. Thus, our interface still permits non-trivial quantum simulation, even as our theorem removes the order-finding advantage within the single-call family we analyzed.

More importantly, the additional readout in $\halfBQPone$ makes weighted normalized unitary traces natural simulation targets. Whether this gives additional computational power over $\complexity{DQC1}$ remains open. Section~\ref{sec:security-simulation} explains a complexity-theoretic barrier to proving such a separation, and Section~\ref{sec:classical-predictions} develops a candidate for separation that is also physically relevant.

\subsection{A barrier to separation}
\label{sec:security-simulation}

The distinction between unweighted and weighted normalized unitary trace estimation suggests a possible separation between $\complexity{DQC1}$ and $\halfBQPone$. One would hope for at least polynomial separations between the two, even if we can't prove complexity-theoretic separation. However, one has to be careful: it turns out that if the weight is computed by a classical circuit of logarithmic depth, then there is no distinction between these two models.

\begin{theorem}[Barrington separation barrier]
\label{thm:barrington-barrier}
Let $f:\{0,1\}^n \times \{0,1\} \to \{0,1\}$ be computable by a polynomial-size classical circuit of logarithmic depth, and let $p=\Pr[f(x,y)=1]$, where $(x,y)$ is the reply of a $\halfBQPone$ request with circuit $U$. Then there is a single $\complexity{DQC1}$ request, of polynomially larger circuit size, whose output bit $z$ satisfies
\[
 \Pr[z=1]=\tfrac12\,p+\tfrac14 .
\]
\end{theorem}

This also implies, for the decision classes, $
 \complexity{P}=\complexity{NC}^1
 \quad\Longrightarrow\quad
 \complexity{DQC1}=\halfBQPone.$ Consequently, finding a separation between $\halfBQPone$ and $\complexity{DQC1}$ is as hard as finding a separation between $\complexity{P}$ and $\complexity{NC}^1$ which is a longstanding open problem. The proof of the above theorem is presented below.
\begin{proof}
Initialize $X,Y,W$ in the maximally mixed state and $Z$ in
$\ket0$. The mixed bit $y$ in $Y$ serves as a guess for the
first output bit of $U$. Apply $O_f$ (Lemma~\ref{lem:nc1-borrowed-oracle}, Appendix~\ref{sec:rr-attack}) to write $f(x,y)$ into $Z$,
then apply $U$ to $X$.

Next apply $\mathrm{CNOT}_{1\to Y}$, so that $Y$ records whether
the guess disagrees with the first output bit. Controlled on
$Y=1$, swap $Z$ with the first borrowed workspace qubit $W_1$.
This qubit is still independently maximally mixed, since $O_f$
restores its borrowed workspace. Finally, measure $Z$.

The guess is correct with probability $1/2$. Conditioned on
a correct guess, $(x,y)$ has exactly the original
$\halfBQPone$ reply distribution, and $Z$ contains $f(x,y)$.
For an incorrect guess, $Z$ is a fair bit. Thus, if the
original acceptance probability is $p$, the new acceptance
probability is
\[
    p'=\frac12p+\frac14
       =\frac12+\frac12\left(p-\frac12\right).
\]
The acceptance gap decreases by only a factor of two, so estimating $\Pr[z=1]$ to half the target error and rescaling recovers $p$ with a constant-factor increase in the number of requests.
The construction uses one clean qubit and one final measured
bit, and therefore is a $\complexity{DQC1}$ circuit.
\end{proof}

\subsection{Testing classical predictions of quantum spin dynamics}
\label{sec:classical-predictions}
By Theorem~\ref{thm:barrington-barrier}, a separation requires weight functions that are unlikely to be computable in logarithmic depth. We therefore consider physical problems whose observables naturally take a weighted form with such weights. Consider the following candidate problem: \emph{Given the initial spin configuration, how well can a classical model predict the spin measured at a particular site after quantum evolution?} Related questions arise in spin transport~\cite{shi2024hydrodynamics} and quantum domain coarsening~\cite{manovitz2025coarsening}.

Fix an efficiently implementable quantum evolution $U(t)$ and a site $i$. Let $f(x)\in\{+1,-1\}$ be a classical model's prediction for the measured spin, computed in polynomial time from the initial configuration $x$. In each $\halfBQPone$ run, the server prepares a uniformly random $\ket{x}$, evolves it with $U(t)$, and measures $Z_i$. The client receives $x$ and the measured spin $S_i=\pm1$, computes $f(x)$ classically, and scores their agreement by $f(x)S_i$. Averaging gives
\begin{equation}
\label{eq:classical-quantum-prediction}
\begin{aligned}
 \mathbb{E}[f(X)S_i]
 &=\frac{1}{2^n}\sum_x f(x)
   \bra{x}U(t)^\dagger Z_iU(t)\ket{x},\\
 \Pr[S_i=f(X)]&=\frac{1+\mathbb{E}[f(X)S_i]}{2}.
\end{aligned}
\end{equation}
The second line gives the prediction's agreement rate, averaged over initial configurations and measurement outcomes. Taking $f(x)=(-1)^{x_i}$ recovers the usual spin autocorrelation. The bounded scores give additive error $\varepsilon$ with $O(\varepsilon^{-2})$ runs at fixed confidence, and the same records can test several classical models.

For a candidate separation, we seek prediction functions $f$ in $\complexity{P}$ that are plausibly outside $\complexity{NC}^1$, beyond the construction of Theorem~\ref{thm:barrington-barrier}. A concrete family comes from Moore's three-dimensional \emph{Majority-Vote Cellular Automata Prediction} problem~\cite[Theorem~1]{moore1997majority}. On a cubic-lattice box of $n$ spins, with exterior spins held at $-1$, each synchronous update replaces every spin by the majority of itself and its six nearest neighbors. Given the initial configuration $x$, a target site $i$, and a time $\tau\le n$, the problem asks for the resulting spin at $i$; fixing $i$ and $\tau$ defines our classical weight $f(x)$. This prediction problem is $\complexity{P}$-complete, so its prediction functions cannot all lie in $\complexity{NC}^1$ assuming $\complexity{P}\ne\complexity{NC}$. These classical predictions provide natural weights for candidate weighted-trace separations between $\halfBQPone$ and $\complexity{DQC1}$. A separation would still require showing that estimating the weighted average over uniformly random $x$ is hard for $\complexity{DQC1}$; the worst-case hardness of evaluating $f$ alone does not establish this.

\section{A finite-size regime without input randomization}
\label{sec:finite-size}

Section~\ref{sec:security} explained why a server that prepares clean inputs and releases one designated bit is protected against every known factoring algorithm run out-of-the-box, and why the only known escape is to run the decoder coherently so that the released bit is a bit of the factor. Randomizing the input register is one way to prevent implementation of such a decoder on an asymptotic scale. However, at finite-size regimes, it may be that the polynomial overhead in implementing such a decoder coherently is itself enough to disallow factoring. Here we consider such an argument.

We first consider a resource estimate for coherent continued-fraction decoder. Such a decoder runs the extended Euclidean algorithm on the $2n$-bit Fourier sample $y$ and $2^{2n}$, tracking the denominators of the convergents until one would exceed $N$. Its registers are the Fourier register itself and a second remainder register ($4n$ qubits), two convergent denominators bounded by $N$ ($2n$), a quotient register ($2n$) and the scratch for a $2n$-bit division ($2n$ to $4n$), plus the output bit and a constant number of control flags. The quotient of each Euclidean step is uncomputed from the updated denominators, $a_k=\lfloor q_{k+1}/q_k\rfloor$, so no computation history accumulates over the $O(n)$ steps and the same registers are reused throughout~\cite[Sec.~5]{proos2003shor}. The total is $10n$ to $12n$ logical qubits, or roughly $20{,}000$ to $25{,}000$ at RSA-2048. The quantum stage that produces $y$ is budgeted at a few thousand logical qubits or fewer in the most qubit-efficient current proposals~\cite{gidney2021factoring,gidney2025factoring}, so the decoder, not the arithmetic, sets the size of a one-bit-readout attack. Thus, against attacks that are not further optimized, a machine with ten thousand logical qubits that releases one output bit may be as secure against factoring as a machine with one thousand logical qubits that releases full output samples (of course, unlike Shor's algorithm, optimizing the clean qubit usage of continued fractions is not a previously studied problem, so there may be great scope for simple optimizations we have not considered).

There is compelling evidence that this workspace is necessary.
By Barrington's theorem, a function has logarithmic-depth circuits exactly when a polynomial-length reversible program on constant workspace, driven by read-only input bits, computes it~\cite{barrington1989}. Theorem~\ref{thm:rr-attack} is the forward direction applied to factoring: shallow rational reconstruction gives a coherent decoder on $O(\log n)$ clean qubits, and on one after the Shor--Jordan reduction. The reverse direction says that a constant-workspace coherent decoder of this form would put rational reconstruction in $\complexity{NC}^1$. The decades of failed attempts to parallelize rational reconstruction are therefore also decades of failed attempts to decode with classical reversible programs on constant workspace. This suggests that there may always be some non-trivial overhead in coherently implementing the decoder.

\section{Discussion}

An apparent obstacle to protecting publicly accessible quantum
computers is obfuscation: an untrusted user can disguise an
algorithm's purpose or distribute it across requests. Our approach
constrains the computational power available through a security
interface. If the full interaction, including adaptive requests and
classical post-processing, lies in a complexity class that excludes
a task, no permitted algorithm can solve it efficiently, however it
is disguised. Under the stated complexity-theoretic assumptions,
this protection persists as quantum computers scale up.

Our work provides a security motivation for using restricted models like
$\halfBQPone$ or $\complexity{DQC1}$ on fully fault-tolerant quantum computers.
Such models have often been studied for their theoretical
interest or as models of computation with limited
resources~\cite{knill1998}. Here, restricting access is a
deliberate choice even when the computer can implement
arbitrary quantum computations. Thus, we provide motivation
for further study of useful quantum applications in
$\halfBQPone$, and of cryptographic and
complexity-theoretic arguments against efficient factoring
through this interface.

\subsection{Future Directions}

We pose a broader question for theoretical computer science:
under what conditions can a powerful server, capable of solving problems
in a complexity class $A$, interact with a less powerful client, whose
local computation is restricted to a class $B$, so as to give the client
full access to an intermediate class $C$, where
$B \subsetneq C \subsetneq A$? This prevents even a malicious client from using the interaction to solve problems outside $C$.
This question connects to the study of knowledge complexity, which
investigates the computational advantage gained through
interaction~\cite{GMR89,GP99}, and to simulation-based security, which
formalizes the requirement that an adversary obtain no capabilities
beyond those provided by a prescribed ideal
functionality~\cite{Canetti01}.
Here, the objective is to characterize the computational power made
available to the client, with both a guarantee of access to $C$ and a
matching upper bound against malicious behavior.

This suggests studying other intermediate complexity classes
as targets for security interfaces. The challenge is to identify
useful tasks they retain, cryptanalytic tasks they exclude, and
preparation, control, and readout constraints that enforce these
limits even under adversarial adaptive use.\\

Deliberate noise injection offers a complementary direction. We can ask how a
specified noise model restricts the accessible computational
power, and whether it suppresses cryptanalytic uses while
retaining useful simulation. Any such guarantee must account
for the user's ability to correct or mitigate the noise.
Understanding these restrictions provides an alternative to
seeking finer intermediate models of computations, or even complexity classes.

Finally, we have focused on only one cryptographic risk posed by quantum computers: attacks on factoring-based cryptography. Other primitives, including Diffie–Hellman key exchange and elliptic-curve digital signatures, are vulnerable to quantum algorithms for discrete logarithms~\cite{shor1997,nist2024transition}. Beyond cryptography, potential risks include the misuse of quantum simulation to support harmful chemical design and weapons development~\cite[Box~9]{oecd2025quantum}. Finding security interfaces, or complementary strategies, that mitigate these risks while preserving beneficial applications is a crucial and promising direction for future inquiry.

\section{Acknowledgments} We thank Dolev Bluvstein and Maddie Cain for helpful discussions that motivated this work.
S.A.Y. thanks Yifan Hong, Mauro Morales, Shrigyan Brahmachari, Samuel Tan, Emil Khabiboulline and Kishor Bharthi for useful discussions.

\par
\onecolumngrid
\vspace{12pt}
\bibliographystyle{apsrev4-2}
\bibliography{references}
\onecolumngrid
\clearpage
\begin{center}
\Large Appendix
\end{center}
\appendix
\onecolumngrid

\section{Classical simulation of one-bit readout}
\label{sec:readout-sampler}

We prove Theorem~\ref{thm:query-class-sampler} by expressing the
measured bit's expectation as an efficiently sampled average of
classically estimable quantities. This avoids simulating the full
output state of $A=\mathcal N O_\pi V$.

\emph{Tractable states and monomial operators.}
A state $\ket\phi$ is \emph{computationally tractable} (CT) if
classical algorithms can sample its basis distribution to total
variation error $2^{-b}$ and evaluate its amplitudes to additive
error $2^{-b}$ in time polynomial in its description length and $b$,
using a consistent global phase~\cite{vandenNest2011probabilistic}.
A tractable preparation supplies these algorithms for $V\ket s$
uniformly in the basis input $s$.

An efficiently invertible basis permutation preserves tractability:
samples are relabelled by $\pi$, and the amplitude at $h$ becomes
$\phi_{\pi^{-1}(h)}$. Products of CT states are also CT.
Normalizer preparations and products of explicit
$O(\log m)$-qubit circuits are tractable~\cite{normalizer}.

A \emph{monomial unitary} permutes basis states and attaches phases:
$P\ket g=\omega(g)\ket{\sigma(g)}$, where $|\omega(g)|=1$.
We require efficient evaluation of $\sigma$, $\sigma^{-1}$,
and $\omega$. Its expectation on a CT state is efficiently
estimable by Van den Nest's sparse-operator theorem
\cite[Theorem~3]{vandenNest2011probabilistic}. The next lemma
combines even exponentially many such expectations.

\begin{lemma}[Sampling a sum of monomial expectations]
\label{lem:ct-monomial}
Let $\ket\phi$ be CT and let
$H=\sum_j c_jP_j$ be Hermitian, where the $P_j$ are efficiently
computable monomial unitaries. Suppose there is an efficiently
sampleable distribution $q$ with $q_j>0$ whenever $c_j\ne0$, and
the weights $c_j/q_j$ are efficiently computable and satisfy
\[
 \frac{|c_j|}{q_j}\le\Lambda,\qquad \Lambda\ge1.
\]
Then $\bra\phi H\ket\phi$ is classically estimable to additive
error $\varepsilon$, with failure probability at most $\zeta$, in
time polynomial in the description length, $\Lambda$,
$1/\varepsilon$, and $\log(1/\zeta)$.
\end{lemma}

\begin{proof}
Write $\mu_j=\bra\phi P_j\ket\phi$. Sampling $j\sim q$ gives
\[
 \bra\phi H\ket\phi
 =\mathbb E_{j\sim q}
   \operatorname{Re}\!\left(\frac{c_j}{q_j}\mu_j\right).
\]
Since $|\mu_j|\le1$, each summand is bounded by $\Lambda$.
The Hoeffding bound requires only
$O(\Lambda^2\varepsilon^{-2}\log(1/\zeta))$ samples.
Estimate each sampled $\mu_j$ to accuracy
$O(\varepsilon/\Lambda)$ using the finite-precision CT theorem,
with failure probabilities summing to at most $\zeta/2$.
Round the weights to accuracy $O(\varepsilon)$ and allocate the
remaining error to the outer average. Every step takes polynomial
time.
\end{proof}

\emph{The Fourier expansion of one bit.}
Let $G=\prod_i\mathbb Z_{2^{t_i}}$, with $\sum_i t_i=m$.
Put the measured register first and let $\rho$ be the measured
digit, counted from the least significant digit $0$.
For $M=2^{\rho+1}$, its $Z$ observable is $+1$ on the first half
of the residues $g_1\bmod M$ and $-1$ on the second half.

\begin{lemma}[Fourier coefficients of a binary digit]
\label{lem:digit-characters}
For $\chi_\ell(g)=e^{2\pi\mathrm{i}\ell g_1/M}$ and
$D_\chi\ket g=\chi(g)\ket g$,
\[
 Z_1=\sum_{\substack{0<\ell<M\\\ell\ \mathrm{odd}}}
       c_\ell D_{\chi_\ell},
 \qquad
 c_\ell=\frac{4}{M(1-e^{-2\pi\mathrm{i}\ell/M})},
 \qquad
 |c_\ell|\le\frac1{\min(\ell,M-\ell)}.
\]
\end{lemma}

\begin{proof}
The constant coefficient is zero. For $\ell\ne0$, the coefficient
is the difference of the two geometric sums over the two halves
of the residues:
\[
 c_\ell=\frac1M\sum_{x=0}^{M-1}
 (-1)^{\lfloor2x/M\rfloor}e^{-2\pi\mathrm{i}\ell x/M}
 =\frac{(1-(-1)^\ell)^2}
        {M(1-e^{-2\pi\mathrm{i}\ell/M})}.
\]
Only odd $\ell$ contribute. The bound follows from
$|1-e^{-2\pi\mathrm{i}\ell/M}|
 =2\sin(\pi\min(\ell,M-\ell)/M)
 \ge4\min(\ell,M-\ell)/M$.
\end{proof}

\begin{proof}[Proof of Theorem~\ref{thm:query-class-sampler}]
Fix any basis input $s$. The state
$\ket{\phi_s}=O_\pi V\ket s$ is CT by the closure property above.
Normalizer conjugation takes a character observable to an
efficiently computable monomial unitary,
$P_\ell=\mathcal N^\dagger D_{\chi_\ell}\mathcal N$
\cite[Theorems~2, 7--8 and Corollary~4]{normalizer}.
The measured bit therefore has bias
\[
 \beta(s)=\bra sA^\dagger Z_1A\ket s
 =\operatorname{Re}\sum_{\ell\ \mathrm{odd}}
 c_\ell\bra{\phi_s}P_\ell\ket{\phi_s}.
\]

Although this sum can have exponentially many terms, we can sample
them with polynomially bounded weights. Group the odd indices into
bands
\[
 \mathcal L_a=
 \{\ell:0<\ell<M,\ \ell\text{ odd},\quad
             2^a\le\min(\ell,M-\ell)<2^{a+1}\}.
\]
Let $J\le\log_2M\le m$ be the number of nonempty bands.
Choose a band uniformly and then choose an index uniformly within
it. Each band consists of at most two intervals of odd integers,
so this sampling and the computation of its size are efficient.
For $\ell\in\mathcal L_a$,
\[
 q_\ell=\frac1{J|\mathcal L_a|},
 \qquad
 \frac{|c_\ell|}{q_\ell}
 \le J|\mathcal L_a|2^{-a}\le2J\le2m.
\]
The weights are computable to the required accuracy with
polynomially many bits, since $M$ has only $O(m)$ bits.
Lemma~\ref{lem:ct-monomial} now estimates $\beta(s)$ in polynomial
time, uniformly in $s$.

To obtain a sample, estimate the bias to error $\varepsilon/2$ with
failure probability at most $\varepsilon/4$, and clip the result
$\widehat\beta$ to $[-1,1]$. Then
$\mathbb E|\widehat\beta-\beta(s)|\le\varepsilon$.
Output $1$ with probability $(1-\widehat\beta)/2$, rounded to
accuracy $\varepsilon/4$. The total variation error is at most
$\frac12\mathbb E|\widehat\beta-\beta(s)|+\varepsilon/4
\le\varepsilon$.

To simulate a $\halfBQPone$ reply, first draw $s$ uniformly and
return it together with this bit. The joint error is the average
of the conditional errors. For $T$ adaptive requests, use error
$\delta/T$ per reply. Conditional on matching transcripts so far,
the real and simulated clients make the same next request; couple
the next replies with disagreement probability at most $\delta/T$.
The probability of any disagreement is at most $\delta$.
\end{proof}

The proof holds for every basis input. An implementation error
$\|\widetilde A-A\|\le\eta$ on the full Hilbert space adds at
most $\eta$ to the reply's total variation error.

\section{Factoring stages as single-call circuits}
\label{sec:readout-arithmetic}

We now identify the preparations, arithmetic operations, and
postprocessing to which Theorem~\ref{thm:query-class-sampler}
applies. The classical decoding of the measurement results occurs
after these quantum stages.

\emph{Arithmetic and workspace.}
An ideal oracle that writes a modular product into a target,
\[
 \ket{x,v}\longmapsto
 \ket{x,\,v\oplus(\prod_i a_i^{x_i}\bmod N)},
\]
is an efficiently computable basis permutation and is its own
inverse. The theorem allows this oracle whether or not it has a
small-workspace quantum implementation. It also allows any actual
classical reversible arithmetic network: its action on all wires,
including arbitrarily initialized workspace, is a permutation
whose inverse is evaluated by reversing the network.
A general quantum implementation specified only on initialized
workspace does not automatically satisfy this hypothesis.

\emph{Standard quantum stages.}
The relevant instances of $A=\mathcal N O_\pi V$ are as follows.
\begin{itemize}[leftmargin=*]
\item \emph{Shor and Jacobs--Mehraban.}
The exponent register is prepared using Hadamards or $F_Q$,
the arithmetic computes $a^x\bmod N$, and the final operation
on that register is a cyclic Fourier transform, as in
Eq.~\eqref{eq:query-jacobs-mehraban}.

\item \emph{Eker{\aa}--H{\aa}stad.}
Two exponent registers control $g^u h^{-v}\bmod N$ and are
Fourier transformed. Their uniform preparations can be written
using $F_{Q_1}\otimes F_{Q_2}$; the lattice decoding is classical
\cite[Sec.~4.3]{ekera2017}.

\item \emph{Kitaev.}
For a fixed probe schedule, defer the probe measurements and move
their final single-qubit gates past the other probes' controlled
multiplications. The resulting circuit is
\[
 \left(\prod_i H_iR_i^{e_i}\right)O
 \left(\prod_i H_i\right),
 \qquad R=\operatorname{diag}(1,\mathrm{i}),
\]
where, for $y<N$,
$O\ket{\mathbf c,y}
 =\ket{\mathbf c,\,a^{\sum_i c_i2^{j_i}}y\bmod N}$,
and $O$ is the identity for $y\ge N$.
Here $\gcd(a,N)=1$, so the arithmetic has an efficient inverse,
and the surrounding circuits are Clifford
\cite[Secs.~3--4]{kitaev1995}.

\item \emph{Regev and its space-optimized variants.}
The preparation is a product discrete Gaussian, followed by one
modular product and $F_D^{\otimes d}$
\cite[Sec.~3]{regev2024efficient}.
The lemma below gives a tractable realization of the preparation.
The same form covers the arithmetic and classical-postprocessing
modifications in Refs.~\cite{ragavan2025space,ekeraGartner2024extensions}.
\end{itemize}

Hadamards, cyclic Fourier transforms, and the indicated Clifford
circuits are normalizers over the appropriate product group.
Their action on basis states is tractable. A semiclassical Fourier
transform with its prescribed feed-forward schedule has the same
measurement statistics as the coherent Fourier transform
\cite{griffiths1996semiclassical}. These observations cover the
specified schedules, not arbitrary adaptive operations within a
request.

\begin{lemma}[Tractable Gaussian preparation]
\label{lem:readout-gaussian}
Let $D=2^t\le4\sqrt d\,R$ and $0<\eta<1$, and set
\[
 k=\min\!\left\{t,
       \left\lceil\log_2(128\pi d^2/\eta)\right\rceil\right\}.
\]
There is an explicit ancilla-free preparation
$V=\bigl(V_\Gamma\otimes H^{\otimes(t-k)}\bigr)^{\otimes d}$ whose
output on the all-zero input is within Euclidean distance $\eta$
of the normalized state with amplitudes
\[
 \exp\!\left[-\frac{\pi}{R^2}
                \sum_{a=1}^d(x_a-D/2)^2\right],
 \qquad x_a\in\{0,\ldots,D-1\}.
\]
Here $V_\Gamma$ acts on $k$ qubits and has size polynomial in $d$ and
$1/\eta$. For inverse-polynomial $\eta$, this preparation is
tractable on every basis input.
\end{lemma}

\begin{proof}
The Gaussian varies slowly enough that only its leading bits need
nonuniform amplitudes. To make this precise, let
\[
 g(x)=e^{-\pi(x-D/2)^2/R^2},\qquad
 B=D/2^k,\qquad
 \widetilde g(x)=g(B\lfloor x/B\rfloor).
\]
The function $\widetilde g$ is constant on blocks of $B$ consecutive
integers. If $k=t$, it equals $g$ exactly. Otherwise,
\[
 |\log g(x)-\log\widetilde g(x)|
 \le\frac{\pi BD}{R^2}\le\frac{16\pi d}{2^k}.
\]
For the product over $d$ coordinates, the ratios of the
unnormalized amplitudes therefore lie between $e^{-\alpha}$ and
$e^\alpha$, where $\alpha=16\pi d^2/2^k\le\eta/8$.
The normalized vectors are therefore within Euclidean distance
$2(e^\alpha-1)<\eta/2$.

The normalized block-constant state is
\[
 \left(\ket\Gamma\otimes
             \ket+^{\otimes(t-k)}\right)^{\otimes d},
 \qquad
 \ket\Gamma\ \propto
 \sum_{c=0}^{2^k-1}g(Bc)\ket c.
\]
There are only $2^k=O(d^2/\eta)$ amplitudes to specify.
Standard state preparation by controlled rotations constructs
$V_\Gamma\ket0=\ket\Gamma$ without extra qubits; synthesize the rotations
so that the total additional state error is at most $\eta/2$.
For inverse-polynomial $\eta$, each explicit block has
$k=O(\log m)$ qubits, where $m$ is the total circuit width.
Thus, on \emph{any} basis input, every block can be simulated
directly, and the full product state is CT.
\end{proof}

This lemma specifies a preparation whose action on random inputs
is tractable. Preparing the correct Gaussian on the all-zero input
alone would not guarantee that property. On initialized inputs,
the replacement changes each run's output law by at most $\eta$;
over $T$ runs, choosing $\eta\le\delta/T$ limits the total change
to $\delta$.

\section{A factoring attack from shallow rational recovery}
\label{sec:rr-attack}

We prove Theorem~\ref{thm:rr-attack}. The construction evaluates a
predicate of an eigenphase using only logarithmically many clean
qubits, which a reduction of Shor and Jordan then trades for a
single clean qubit. It never stores the full phase estimate: whenever the
classical decoder needs an input bit, the circuit computes that
bit, uses it, and uncomputes it.

\begin{lemma}[$\mathsf{NC}^1$ evaluation with borrowed workspace]
\label{lem:nc1-borrowed-oracle}
Suppose each output bit of $f:\{0,1\}^n\to\{0,1\}^k$ has a
supplied polynomial-size, logarithmic-depth, fan-in-two Boolean
circuit, with $k=\operatorname{poly}(n)$. Then a polynomial-size
reversible circuit implements
\begin{equation}
 O_f\ket{x,w,\kappa}=\ket{x,w\oplus f(x),\kappa}
 \label{eq:nc1-borrowed-oracle}
\end{equation}
for every $x,w$ and every pair of borrowed bits
$\kappa\in\{0,1\}^2$. The borrowed bits need no initialization
and are restored at the end.
\end{lemma}

\begin{proof}
A permutation program is a sequence of fixed permutations, each
selected by one unchanged input bit. Barrington's theorem converts
an $\mathsf{NC}^1$ circuit into a polynomial-length program of
constant width. We use its finite-simple-group form
\cite[Theorem~1]{ablayev2004branching} with the alternating group
$A_8$ on the eight states of one target bit and two borrowed bits.
The desired permutation $X\otimes I_4$ is a product of four
transpositions, so it is a nonidentity element of $A_8$.
The theorem supplies a program whose total action is the identity
when $f_j(x)=0$ and $X\otimes I_4$ when $f_j(x)=1$.
Each instruction is a constant-size controlled permutation.
Apply one such program to each target bit $w_j$, reusing the same
borrowed bits. The resulting identity holds on every basis state,
and therefore on superpositions as well.
\end{proof}

For example, modular exponentiation with public base and modulus
has the required shallow circuits after classical preprocessing
\cite[Claim~2]{jacobs2024half}. In the argument below, the
assumption $\mathrm{RR}\in\mathsf{NC}^1$ supplies such circuits
for the indexed output bits of rational recovery, as defined in
the main text.

\begin{proof}[Proof of Theorem~\ref{thm:rr-attack}]
Let $N$ be a product of two distinct odd primes. Fix a unit $a$
modulo $N$, write $n=\lceil\log_2N\rceil$, and let
$r=\operatorname{ord}_N(a)$. We first recover $r$ using one-bit
threshold tests; the usual randomized reduction then factors $N$.

\emph{An efficiently accessible eigenphase.}
On two $n$-qubit registers, use the basis permutation
\[
 B_a(x,y)=
 \begin{cases}
 (ax\bmod N,\ a^{-1}y\bmod N),&x,y<N,\\
 (x,y),&\text{otherwise}.
 \end{cases}
\]
Controlled implementations use polynomially many gates,
borrowed workspace, and $O(\log n)$ clean qubits. On valid
residues this is the modular bimultiplication construction of
Ref.~\cite[Secs.~2.2--2.3]{gidney2017dirty}: perform
$y\mathrel{+}=ax$, $x\mathrel{-}=a^{-1}y$, and
$y\mathrel{+}=ax$ modulo $N$, then swap the registers and negate
the second. To enforce the identity outside the valid region,
compute the flag $[x<N\text{ and }y<N]$, control the construction
on it, and uncompute the flag. Comparisons use
Lemma~\ref{lem:nc1-borrowed-oracle}, and the flag can be erased
because the valid region is preserved. Finally,
$B_a^h=B_{a^h}$, so the powers needed below are implemented by
precomputing new public constants, rather than repeating $B_a$
$h$ times.

The maximally mixed target can be regarded as a uniform mixture
of eigenstates~\cite{ParkerPlenio2000}. Every eigenphase
$\theta\in[0,1)$ has a reduced denominator $d$ dividing $r$,
because $B_a^r=I$. Moreover, whenever $x$ is a unit modulo $N$
and $y<N$, the pair $(x,y)$ lies on a cycle of length exactly $r$.
Each such cycle supplies the phases $j/r$ uniformly, of which
a fraction $\varphi(r)/r$ have reduced denominator $r$. Hence
\begin{equation}
 \Pr(d=r)\ge
 \frac{N\varphi(N)}{2^{2n}}\frac{\varphi(r)}r
 \ge\gamma_n:=\frac{2}{15(n+1)}.
 \label{eq:rr-full-denominator}
\end{equation}
Here $\varphi(N)/N\ge8/15$ and $N/2^n\ge1/2$.
For the remaining bound, if $r$ has $h$ distinct prime factors,
the $i$th smallest is at least $i+1$, so
$\varphi(r)/r\ge\prod_{i=1}^h i/(i+1)
=1/(h+1)\ge1/(n+1)$.

\emph{A shallow predicate of the phase.}
Set $\mu=2n+2$ and choose a classical dyadic shift $\tau$.
For the moment, imagine having the $\mu$-bit integer
\[
 \xi=\left\lfloor2^\mu((\theta+\tau)\bmod1)\right\rfloor.
\]
Undo the shift by defining
\[
 \vartheta_\tau(\xi)=(\xi/2^\mu-\tau)\bmod1,\qquad
 t_\tau(\xi)=
 \begin{cases}
 0,&\vartheta_\tau(\xi)>1-2^{-\mu},\\
 \vartheta_\tau(\xi),&\text{otherwise}.
 \end{cases}
\]
The endpoint correction accounts for a phase near zero wrapping
around to a number near one. Every nonzero eigenphase is at least
$1/r$ from both endpoints, and $r<N$, so
\[
 |t_\tau(\xi)-\theta|<2^{-\mu}<\frac1{2N^2}.
\]
Consequently $\mathrm{RR}(N,t_\tau(\xi))=d$. For a threshold $K$, the
predicate
\[
 f_{K,\tau}(\xi)=\mathbf1\{\mathrm{RR}(N,t_\tau(\xi))\ge K\}
\]
is in $\mathsf{NC}^1$: rational recovery is shallow by
assumption, and the dyadic arithmetic, comparisons, and endpoint
correction are also shallow. By
Lemma~\ref{lem:nc1-borrowed-oracle}, a polynomial-length program
can write this predicate into one initially zero output bit,
using two borrowed bits. Each instruction needs only one bit of
$\xi$, which is why the full integer need not be stored.

\emph{Supplying one phase bit at a time.}
Number the bits of $\xi$ from most to least significant. Its $j$th
bit specifies which half of the unit circle contains
$2^{j-1}(\theta+\tau)\bmod1$. The only difficulty is that a phase
very close to $0$ or $1/2$ cannot be classified reliably at low
precision. A random shift makes such cases rare.

Choose $\tau$ uniformly from the dyadic grid with $\mu+b$ bits.
For every fixed $\theta$, a union bound gives
\[
 \Pr_\tau\!\left[
   \text{some }2^{j-1}(\theta+\tau)\bmod1
   \text{ lies within }\Delta\text{ of }0\text{ or }1/2
 \right]
 =O\!\left(\mu(\Delta+2^{-b})\right).
\]
Choose $\Delta$ inverse polynomially small and
$b=O(\log n)$ so that this probability is at most
$\gamma_n/1000$.
Away from these boundaries, coherent phase estimation of
\[
 e^{2\pi\mathrm{i}2^{j-1}\tau}B_a^{\,2^{j-1}}
\]
supplies the required bit with error probability at most
$\varepsilon$, using
$O(\log(1/(\Delta\varepsilon)))$ clean qubits
\cite[Sec.~5 and Appendix~C]{CleveEtAl1998}.

For each program instruction, compute this bit, use it to control
the instruction's permutation, and reverse the phase estimation.
If the bit were exact, this would restore the workspace exactly.
With wrong-bit probability at most $\varepsilon$, the resulting
state differs from that ideal action by at most
$2\sqrt\varepsilon$ in Euclidean norm. For a program of length
$L=\operatorname{poly}(n)$, replacing the instructions one at
a time bounds the total error by $2L\sqrt\varepsilon$.
Choose $\varepsilon$ so this is at most $\gamma_n/1000$,
and give gate synthesis the same error budget. All required
precisions are polynomial, the number of clean qubits remains
$O(\log n)$, and no intermediate measurement is used.
Including the bad shifts, the acceptance probability differs
from the ideal one by less than $\gamma_n/100$, averaged over
$\tau$ and the mixed target.

\emph{Recovering the order from one-bit tests.}
The ideal acceptance probability is
\[
 p_K=\Pr(d\ge K),\qquad
 \begin{cases}
 p_K=0,&K>r,\\
 p_K\ge\gamma_n,&K\le r.
 \end{cases}
\]
Thus an additive estimate with error much smaller than
$\gamma_n$ distinguishes whether $K\le r$.
The circuit above uses $O(\log n)$ clean qubits and one measured
bit. Shor and Jordan's reduction~\cite[Sec.~1]{ShorJordan2008}
estimates the acceptance probability of such a circuit to
inverse-polynomial additive accuracy from polynomially many
$\complexity{DQC1}$ requests, each with a single clean qubit;
averaging over the classical shift $\tau$ preserves the estimate.
(Through $\halfBQPone$ directly, one may instead designate the
$q=O(\log n)$ clean input bits in advance and retain only replies
in which their revealed values are zero. This event has probability
$2^{-q}=1/\operatorname{poly}(n)$, conditioned on it the other
inputs remain maximally mixed, and capping the number of attempts
adds an arbitrarily small failure probability at polynomial cost.)
Repeating with fresh shifts and estimating the mean gives the
required threshold test.

Binary search over $1\le r<N$ uses $O(n)$ such tests.
Amplify each test to failure probability at most $1/(10n)$
to recover $r$ with bounded error. Applying this procedure to
random choices of $a$, with classical gcd computations and
verification, gives the standard order-to-factor reduction
\cite{shor1997}. This places $\Factoring$ in
$\complexity{BPP}^{\complexity{SampDQC1}}$; the containment in
$\complexity{BPP}^{\complexity{Samp}\halfBQPone}$ is
Theorem~\ref{thm:containments}.
\end{proof}

The circuits of this proof repeatedly alternate phase estimation
and arithmetic, so they lie outside the single-call family of
Theorem~\ref{thm:query-class-sampler}, as they must.

\end{document}